%% file: arxiv.tex
\documentclass[conference]{IEEEtran}
\usepackage{graphicx,times,cite,amsmath, amssymb, epsfig}
\usepackage{multirow}
\usepackage{subfig}
\usepackage[noend]{algorithmic}
\usepackage{algorithm}
\usepackage{color}
\usepackage{stackengine}

\newlength{\figwidth}
\IEEEoverridecommandlockouts 
\newtheorem{theorem}{Theorem}
\def\delequal{\mathrel{\ensurestackMath{\stackon[1pt]{=}{\scriptscriptstyle\Delta}}}}
\begin{document}
\setlength{\pdfpagewidth}{8.5in}
\setlength{\pdfpageheight}{11in}
\title{Opportunistic Wireless Energy Transfer in Point-to-Point Links}
\author{

\IEEEauthorblockN{Amanthi Thudugalage, Saman Atapattu and Jamie Evans\\}
\IEEEauthorblockA{Department of Electrical and Electronic Engineering, University of Melbourne, Victoria, Australia.\\}
\IEEEauthorblockA{Email:a.thudugalage@student.unimelb.edu.au, \{saman.atapattu, jse\}@unimelb.edu.au}}


\maketitle

\begin{abstract}

In this paper we consider wireless energy transfer for a point-to-point link. The energy transmitter sees a finite number of independent channel realizations, and, armed with (causal) knowledge of the channels, must decide how much energy to transmit in each time slot. The objective is to maximize the expected energy transferred to the receiver at the end of the time period. We show that the optimal energy allocation policy is binary: the transmitter sends no energy or all energy in a slot with this decision based on a simple threshold on the channel. As intuition demands, this threshold for transmission decreases as we move closer to the last available time slot. The performance of the optimal scheme is studied both analytically and numerically.

%
\end{abstract}

\section{Introduction}
Wireless sensor networks (WSNs) are widely being used for sensing in smart environments by using emerging wireless communication techniques, e.g., \cite{Guo2018guotcom} and references therein. 
These techniques can be utilized in rescue missions, crowd monitoring, disaster area monitoring and so on. 
In such scenarios, it may be required to gather information from an energy scarce node such as Internet of Things (IoT), backscatter communications,  unmanned aerial vehicles and multi-hop networks operating in a challenging environment using point-to-point links~\cite{Sudheesh2017,Zhao2018coml,Atapattu2019twc}. 
The duration of such communication sessions can be prolonged by incorporating wireless energy transfer (WET) schemes. 
Unlike other renewable energy harvesting (EH) techniques, WET provides predictable and controllable energy which increases reliability~\cite{Xiao2015,AmanthiICC2016}. However, the amount of energy that can be harvested through WET or simultaneous information and power transmission (SWIPT) is relatively low in both point-to-point links and relaying links, mainly due to the path-loss~\cite{Chen2015,Saman2,Atapattu2016twc}. 

With the advancement of low power circuitry and development of radio frequency (RF) energy harvesting circuits, significant research attention has been drawn to networks with WET~\cite{Huang2014}
In~\cite{Liu2014}, a harvest-then-transmit protocol is introduced in which the access point first broadcasts wireless power to all users via energy beamforming in the downlink (DL), and the users then use harvested energy to send their independent information to the access point simultaneously in the uplink (UL). In order to guarantee rate fairness, the minimum throughput is maximized by joint design of DL-UL time allocation, the DL energy beamforming, the UL transmit power allocation, and the receive beamforming. 

Various opportunistic energy transfer schemes are considered in the literature for relay networks~\cite{Wang2015} and cognitive radio networks~\cite{Wu2015,Zhai2017}. However, point-to-point link  is the focus of this paper. Such a network is considered in \cite{Ozel2011}, and two objectives, namely  i) maximizing the throughput by a deadline; and ii) minimizing the transmission completion time, are considered. Energy allocation over a finite horizon is considered based on channel conditions and time varying energy sources in~\cite{Ho2012}. The throughput is maximized by considering causal and non-causal channel state information (CSI). 
Another design of online transmission strategies for slotted energy harvesting is considered in \cite{ Gomez2017}. This work focuses on minimizing the gap between the maximum rate obtained using offline and online policies. These works mainly focus on maximizing throughput. However, in some applications, the sensor nodes may be highly restricted to perform the EH before a given time (i.e., a time deadline) in order to perform a specific task. In such situations maximizing the harvested energy with time constraints may be of high importance. To the best of our knowledge such a problem has not been studied in the literature.  
Therefore, this paper focuses on finite horizon WET with transmit energy constraints. 

This paper considers a wireless network with an energy-constrained  transmitter and a time-constrained EH user. We study the problem of maximizing the expected total harvested energy  over a finite horizon with causal CSI. We first analyze the problem without channel estimation energy,  and we prove the optimum energy transfer policy to be a threshold policy. We then analyze the optimization problem with channel estimation energy by assuming the  threshold policy. We also analyze the genie-aided scheme as a benchmark. %

The rest of this paper is organized as follows. Section~\ref{sec:SysModel} presents the system model. Section~\ref{sec:Noestimate} solves optimization problems without and with channel estimation energy. Section~\ref{sec:MaxMinResults} presents numerical and simulation results followed by the concluding remark in Section~\ref{sec:Conclusion}.

\section{System Model}
\label{sec:SysModel}

This section describes the  network model and the corresponding analytical model  for a point-to-point SISO wireless network used for EH.  
\subsection{Network Model}
\begin{figure}
  \centering
  \subfloat[]{\label{f:System_Model}\includegraphics[width=1\figwidth]{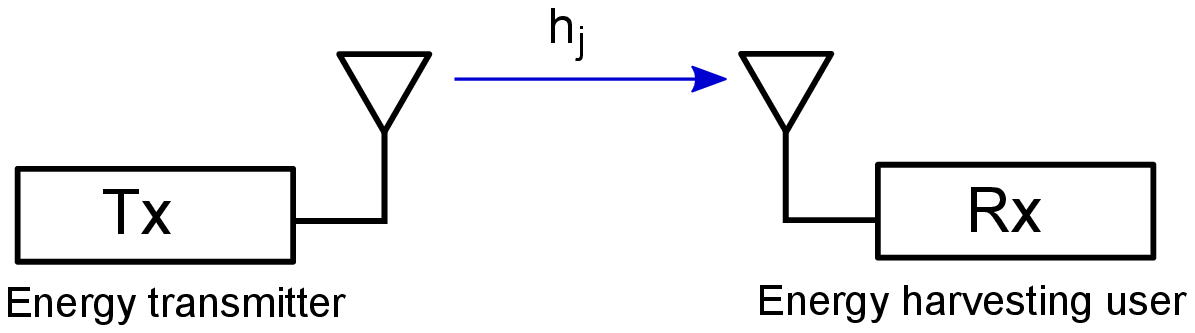}} \\
  \subfloat[]{\label{f:Time_slots}\includegraphics[width=1\figwidth]{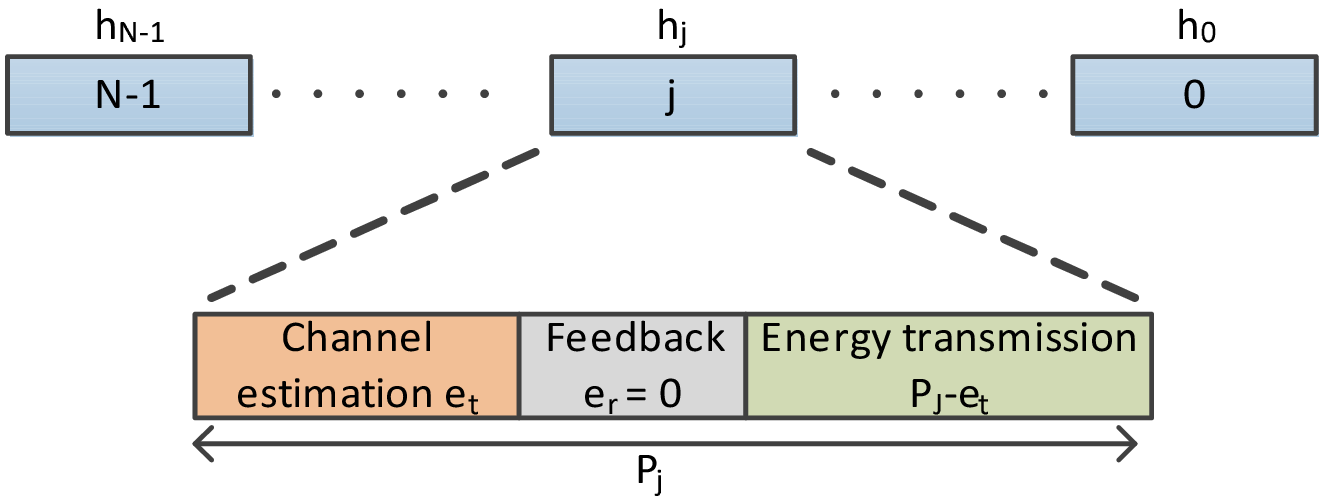}}
\caption{ (a) A point-to-point SISO wireless network for EH; and (b) Frame structure for the energy transfer protocol.}
\end{figure}
We consider a point-to-point SISO wireless network as shown in Fig.~\ref{f:System_Model}.  The network consists of a single-antenna power transmitter denoted as Tx, and a single-antenna EH user denoted as Rx. 
The power transmitter has a fixed energy source with a maximum available energy level $P$. The time varying channel between Tx and Rx is considered to be block faded with $N$ block faded time frames before the deadline. This transmitter can transmit power to the Rx in any of the $N$ frames and it is the precise allocation of the available transmit energy $P$ between the $N$ frames that is the topic of this paper. This means that Tx can transmit all $P$ energy to Rx by using maximum $N$ frames. 
Then, the EH user can harvest energy within those time frames which are indexed from $j=N-1$ to $j=0$ in Fig.~\ref{f:Time_slots}. It is important to note that the indexing of the frames is in the descending order, so that the frame index reveals the number of frames remaining for the future. 
We assume that the distance between Tx and Rx is fixed in the time duration considered. Therefore, the path-loss factor is assumed to be the same for all $N$ frames. Thus, for presentation simplicity, we omit the path-loss factor without loss of generality. The block fading channel between Tx and Rx at the time frame $j$ is $h_j$ which is a complex Gaussian random variable with zero-mean and unit-variance, i.e., $h_{j} \sim \mathcal{CN}(0,1)$. All $h_j$s are independent and identically distributed (i.i.d.) random variables.

 Since we consider an opportunistic wireless energy transfer, the energy transmitter should know the channel $h_j$ before energy transfer to the EH user. We assume perfect channel estimation is carried out at Rx. Since the EH user does not have any energy source, at the beginning of each frame, the energy transmitter transmits $e_t$ amount of energy towards Rx for the channel estimation. We denote the total energy transmitted by Tx in frame $j$ as $\rho_j$, which includes both energy transmitted for channel estimation and harvesting. Using that received signal, the EH user estimates the channel $h_j$, and calculates the transmit energy required for energy transfer within the same frame, $\rho_j - e_t$. Then, the EH user informs this information to Tx during the feedback phase. We assume that the energy associated in the feedback phase is negligible, i.e., $e_r=0$, and the EH user is sufficiently charged to perform the feedback. Then, the energy transmitter transmits, $\rho_j-e_t$ amount of energy for EH at Rx.

\subsection{Analytical Model}

Since the energy level for each set of $N$ frames at Tx is limited to $P$, the available energy at Tx at the beginning of frame $j$ which depends on all the transmit energy of previous frames, can be given as $P_{j,(available)}=P-\sum_{i=j+1}^{N-1} \rho_i$. 
Then, the energy transmitted at the frame $j$, $\rho_j$ also depends on the energy transmitted at the previous frames. Although, this $\rho_j$ may be a function of all previous channel gains, i.e., $\rho_j({h}_{N-1},\cdots,{h}_{j-1} ) $, we use $\rho_j$ for the sake of simplicity. We denote the state of energy availability at frame $j$ as $s_j$, where $s_j=1$ if there is energy available at the frame $j$, i.e., $P-\sum_{i=j+1}^{N-1} s_i \rho_i > 0$, and $s_j=0$ if the energy source has run out. Therefore, $s_j$ can be given as 
\begin{equation}\label{eq:S}
s_j=\left\{
	\begin{array}{ll}
		1,  &  \text{if} ~ P-\sum_{i=j+1}^{N-1} s_i \rho_i > 0  \\
		0, & \text{otherwise.}
	\end{array}
\right.
\end{equation}
The importance of $s_j$ lies in the case of $\rho_j< e_t$ which means that the leftover energy at frame $j$ is not sufficient for the channel estimation. In such situations, the state $s_j$ sets to zero in order to prevent negative value for the harvested energy.

The harvested energy at the frame $j$ can be given as $
 E_j= \eta_j s_j \left(\rho_j - e_t\right)| h_j|^2 + \sigma^2_{n_j}$
where, $\eta_j$ is the energy conversion efficiency, and $\sigma^2_{n_j}$ is the energy associated with received noise. Without loss of generality, we assume perfect energy conversion efficiency, i.e., $\eta_j=1$, and the energy associated with noise is negligible, i.e., $\sigma^2_{n_j}=0$. Then, the expected total harvested energy for $N$ frames can be given as 
\begin{equation}\label{eq:E_sum}
 \tilde{E}_T= \sum_{j=0}^{N-1} \mathbb{E}_{h_j} \Big[ s_j \left(\rho_j - e_t\right) | h_j|^2\Big],
\end{equation}
where $\mathbb{E}_{h_j}(.)$  is the expectation operation with respect to $h_j$. We can achieve maximum harvested energy within $N$ frames by maximizing $ \tilde{E}_T$. This can be done by optimizing the energy transfer at each frame, $\rho_j$ based on the past and present channel knowledge only, i.e., $| h_i |^2$ $i\in[N-1,j]$. However, knowledge of the future channel gains, i.e., $| h_i |^2$ $i\in[j-1,0]$, are unknown. This procedure may be called as {\it opportunistic} energy transfer which is discussed in detail in the next section.

\section{Opportunistic Energy transfer}
\label{sec:Noestimate}
In this section, we maximize the expected total harvested energy over $N$ time frames with the transmit energy budget $P$. In particular, we optimize the amount of energy to be transmitted in each frame, i.e., $\rho_j$, $\forall j \in[0,N-1]$. In general, the corresponding optimization problem can be given as 
\begin{subequations}\label{eq:OptE_sum_et}
\begin{alignat}{4}
 \underset{~{ \rho_j, s_j } }{\mathrm{max}} ~~~ &\tilde{E}_T \label{eq:OptE_sum_et_a}\\
 \text{s.~t.} ~~~
 & \sum_{j=0}^{N-1} s_j \rho_j \leq P, \label{eq:OptE_sum_et_b} \\
 & \rho_j \geq e_t ~~ \forall j=\{j; j \in [0, N-1],  s_j\neq 0\},\label{eq:OptE_sum_et_c}
\end{alignat}
\end{subequations}
where \eqref{eq:OptE_sum_et_b} is for total transmit energy constraint, and \eqref{eq:OptE_sum_et_c} is to ensure that the available energy is adequate for channel estimation when $s_j \neq 0$.

In the following subsections, we solve this optimization problem: i) without channel estimation energy ($e_t=0$), and ii) with channel estimation energy ($e_t \neq 0$). 
\subsection{Without Channel Estimation Energy ($e_t=0$)}
\label{sec:Problem}
When $e_t=0$, the  scenario $\rho_j < e_t$  can be handled by using $\rho_j >0$ constraint. Then, we can use $s_j=1, ~\forall j \in [0, N-1]$ without any harm for the original optimization problem.  Further, the total transmit energy constraint \eqref{eq:OptE_sum_et_b} should be met with an equality at the optimum point because $\rho_j$s can be scaled such that the equality condition is met. Hence, the optimization problem \eqref{eq:OptE_sum_et} can be given as
\begin{equation}\label{eq:OptE_sum2}
\begin{aligned}
 \underset{~{ \rho_j }}{\mathrm{max}} ~~~ & \mathbb{E} \Bigg[{ \displaystyle      \sum_{j=0}^{N-1} \rho_j | h_j|^2   }\Bigg] \\
 \text{s.~t.} ~~~
 & \sum_{j=0}^{N-1} \rho_j = P, ~\rho_j \geq 0, ~~\forall j.
\end{aligned}
\end{equation}
As parameters $\rho_j$s are sequential, we can reformulate \eqref{eq:OptE_sum2} in order to obtain a recursive formula. Let $R_{N-1}(P)$ denote the optimum value of the objective function in~\eqref{eq:OptE_sum2}. At the beginning of frame $N-1$, the optimum energy for that frame, say $\rho^*_{N-1}$, must be calculated first. Then, at the beginning of next frame, i.e., $ N-2$, the available energy at the energy source must be updated as $P-\rho^*_{N-1}$, and the optimum energy, $\rho^*_{N-2}$, must be calculated as previous case. Hence, an equivalent  optimization problem for \eqref{eq:OptE_sum2} can be given as 
\begin{equation}\label{eq:EN(P)}
\begin{aligned}
R_{N-1}(P)= & \underset{~ \rho_{N-1}}{~~\mathrm{max} ~~~}  \mathbb{E}_{h_{N-1}}\Big[\rho_{N-1} | h_{N-1} |^2   \\ & ~~~~~~~~~~~~~~~~~+ R_{N-2}(P-\rho_{N-1}) \Big]\\
 \text{s.~t.} ~~~
 & 0 \leq \rho_{N-1} \leq P.
\end{aligned}
\end{equation}
This problem has a recursive structure, and it has only one decision variable $\rho_{N-1}$, which is the transmit energy at the frame $N-1$. It is worth noting that the function $R_{N-1}(P)$ is a deterministic function as $P$ is deterministic. However, $R_{N-2}(P-\rho_{N-1})$ is a random quantity as the variable, $\rho_{N-1}$  which depends on the random variable $h_{N-1}$ is random. Therefore, the expectation operation, $\mathbb{E}_{h_{N-1}}[.]$  applies to $R_{N-2}(P-\rho_{N-1})$ as well.

\subsubsection{Optimum Solution}
\label{sec:Solution}
Theorem~\ref{t:th_no_et} gives the optimum policy and the optimum harvested energy for the problem in \eqref{eq:EN(P)}.

\begin{theorem}\label{t:th_no_et}
Without channel estimation energy ($e_t=0$), the optimum expected total harvested energy using $N$ frames over block fading Rayleigh channels is 
\begin{equation}\label{eq:Sol}
R_{N-1}(P)=P \sum_{j=0}^{N-1}c_j
\end{equation}
where $c_0=1$ and $c_j= e^{-\sum_{i=0}^{j-1} c_i}$. The optimum transmit energy at frame $j$ is 
\begin{equation}\label{eq:cj}
\rho_j^*=\left\{
	\begin{array}{ll}
		P-\sum_{i=j+1}^{N-1}\rho_i,  &  | h_j|^2 \geq \sum_{i=0}^{j-1}c_i,    \\
		0,  &  \text{otherwise.}  \\
	\end{array}
\right.
\end{equation}
\end{theorem}
 
\begin{proof}
We use method of induction to prove Theorem~\ref{t:th_no_et}. 

Base Case: For one frame ($N=1$) over Rayleigh fading, problem in  \eqref{eq:EN(P)} can be given as
\begin{equation*}\label{eq:E0(P)int}
\begin{aligned}
R_{0}(P)=  \underset{~ \rho_{0}(x)}{\mathrm{max}}  &\int_o^\infty \rho_{0}(x)x e^{-x} dx   \\
 \text{s.~t.} ~~~
 & 0 \leq \rho_{0}(x) \leq P.
\end{aligned}
\end{equation*}
It is obvious that the optimum transmit energy $\rho_0^*(x)$ is $P$, i.e.,  all the available energy is transmitted in the same frame regardless of the channel condition. Therefore, $R_0(P)=P c_0$, where $c_0=\int_0^\infty x e^{-x}dx= 1$ which is the mean of the channel gain. 

Inductive hypothesis: Assume that the solution in ~\eqref{eq:Sol} is true for $N=n$. Then, we have
\begin{equation}\label{eq:Sol1}
R_{n-1}(P)=P \sum_{j=0}^{n-1}c_j.
\end{equation}
with $\rho^*_j=P-\sum_{i=j+1}^{n-1} \rho_i$ if $|h_j|^2 \geq \sum_{j=0}^{j-1} c_i$; and $\rho^*_j=0$ otherwise.

Inductive Step: Consider the case of $N=n+1$. The objective function in \eqref{eq:EN(P)} can be written as
\begin{equation}\label{eq:En+1(P)_optW}
\begin{aligned}
R_{n}(P) = & \mathbb{E}_{h_{n}}\Big[\rho_{n} | h_{n} |^2  +  R_{n-1}(P-\rho_{n}) \Big] \\
 \stackrel{(a)}{=} &  \mathbb{E}_{h_{n}}\left[\rho_{n} | h_{n} |^2  +  (P-\rho_{n}) \sum_{j=0}^{n-1}c_j \right] \\ 
\stackrel{(b)}{=} &\int\limits_0^\infty \rho_{n}(x) \Big(x -\sum_{j=0}^{n-1}c_j \Big) e^{-x} dx +  P \sum_{j=0}^{n-1}c_j  \\
\end{aligned}
\end{equation}
where (a) follows by substituting from \eqref{eq:Sol1},  and (b) follows by  rearranging the terms and using $\mathbb{E}[g(x)]=\int g(x)f(x) dx$.  Here $f(x)$ is the probability density function (PDF) of $x$. Thus, the optimization problem can be given as 
{\small\begin{equation}\label{eq:En(P)int}
\begin{aligned}
R_{n}(P)=  \underset{~ \rho_{n}(x)}{\mathrm{max}}  &\int\limits_0^\infty \rho_{n}(x) \Big(x -\sum_{j=0}^{n-1}c_j \Big) e^{-x} dx +  P \sum_{j=0}^{n-1}c_j  \\
 \text{s.~t.} ~~~
 & 0 \leq \rho_{n}(x) \leq P.
\end{aligned}
\end{equation}}
For maximization, we can note that $\rho_n=0$ if $x < \sum_{j=0}^{n-1}c_j$; and $\rho_n=P$ otherwise. Then, we have
\begin{equation}\label{eq:En+1(P)_optW4}
\begin{aligned}
R_{n}(P)= & P\int\limits_{\sum_{j=0}^{n-1}c_j}^\infty  \Big(x -\sum_{j=0}^{n-1}c_j \Big) e^{-x} dx +  P \sum_{j=0}^{n-1}c_j  \\
= & P e^{-\sum_{j=0}^{n-1}c_j}+  P \sum_{j=0}^{n-1}c_j ~~
= P \sum_{j=0}^{n}c_j
\end{aligned}
\end{equation}
where the last equality follows by substituting  $c_j=e^{ -\sum_{i=0}^{j-1} c_i}$ with $c_0=1$. Based on the rule of induction, Theorem~\ref{t:th_no_et} is true for all $N > 0$. 
This completes the proof.
\end{proof}

Theorem~\ref{t:th_no_et} implies that, if harvestable energy of current frame (say frame $j$), $P| h_j|^2$, is greater than the expected total harvestable energy from the remaining frames, all the available energy $\rho_j^*=P$ will be transmitted at frame $j$. This can be referred to as an all-or-nothing threshold policy, where all the available energy will be transmitted at frame $j$ if the threshold is exceeded, otherwise  remain idle.  This means that the optimum transmit energy at each frame is given by a threshold policy. Specifically, the threshold at frame $j$ is
\begin{equation}\label{eq:thresh_no_et}
\gamma_j=\sum_{i=0}^{j-1} c_i
\end{equation}
 where $c_0=1$ and $c_i=e^{ -\sum_{l=0}^{i-1} c_l}$. 
 
\subsubsection{Genie-aided Energy Transfer}
\label{subsubsec:genie}
For a performance comparison, we consider an energy transfer scheme with non-causal channel state information. This scheme is addressed as the genie-aided energy transfer scheme hereafter. If the energy transmitter knows all channel gains $h_j$s, $~\forall j=0,\cdots,N-1$, energy transmitter may transmit all $P$ energy at the frame with maximum channel gain. We denote that frame index as $j^*$ where $j^*=  \displaystyle{\mathrm{arg}~\mathrm{max}_j\left(|h_{N-1}|^2,\cdots,|h_0|^2 \right)}$. Then, the expected harvested energy using $N$ frames can be given as 
\begin{equation}\label{eq:gen}
\begin{aligned}
\tilde{E}_G &=  P \mathbb{E}_{h_{j^*}}\left[|h_{j^*}|^2\right] = P\int_0^\infty xf_{|h_{j^*}|^2}(x) dx\\
& = P \int_0^\infty x NPe^{-x}(1-e^{-x})^{N-1} = P\sum_{n=1}^{N}\frac{1}{n}.
\end{aligned}
\end{equation}

\subsubsection{Asymptotic Analysis}
\label{subsubsec:SISORayleigh}

In this section, we analyze the asymptotic behavior of the optimum harvested energy for both opportunistic and genie-aided energy transfer schemes, when the number of frames becomes infinitely large, i.e., $N \rightarrow \infty$. Theorem~\ref{t:Asym} gives the maximum expected total harvested energy convergence with $N$.

\begin{theorem}\label{t:Asym}
Without channel estimation energy ($e_t=0$), the optimum expected total harvested energy converges with $N \rightarrow \infty$ as

For the opportunistic energy transfer:
\begin{equation}\label{eq:SISOasym}
\lim_{N\rightarrow \infty} R_{N-1}(P) =P\ln(N)
\end{equation}

For genie-aided energy transfer:
\begin{equation}\label{eq:SISOasym_gen}
\lim_{N\rightarrow \infty} R_{N-1}(P) =P\gamma+P\ln(N)
\end{equation}
where ${\gamma}$ is the Euler-Mascheroni constant (${\gamma}=0.5772...$)\cite{MORTICI2010}.
\end{theorem}
 
\begin{proof}
Define a sequence $U_N$ as $U_N \delequal \sum_{j=0}^N c_j-\ln(N).$
The term $U_{N+1}-U_{N}$ can be give by 
\begin{equation}\label{eq:Un_Yn1}
U_{N+1}-U_N= c_{N+1}-\ln(N+1) +\ln(N).
\end{equation}
By using the mean value theorem, we have $\ln(N+1) -\ln(N)=\frac{1}{N+\theta}, $ where, $0<\theta<1$. Therefore, \eqref{eq:Un_Yn1} can be written as 
\begin{equation}\label{eq:Un_mean}
U_{N+1}-U_N= c_{N+1}-\frac{1}{N+\theta}.
\end{equation}
The term $c_{N+1}=e^{-\sum_{j=0}^{N}c_j}$ can also be given in a recursive formula as $c_{N+1}=c_{N}e^{-c_{N}}$. As $c_0=1$ and $0<e^{-x} \leq 1$ for $x \geq 0$, by using this recursive structure of $c_N$, it can be shown that $c_N$ is a decreasing function which converges to zero, i.e., $\lim_{N\rightarrow \infty} c_{N}=0$. Therefore, 
$\lim_{N \rightarrow \infty}U_{N+1}-U_N= c_{N+1}-\frac{1}{N+\theta}=0$. In other words, the function $U_N$ converges. Furthermore, by computing $U_N$ for large $N$, we can show that the function $U_N$ converges to zero for large $N$. Therefore, for opportunistic energy transfer, the expected harvested energy for large $N$ converges as in \eqref{eq:SISOasym}.

For the genie aided method, the maximum expected total harvested energy is given as $\tilde{E}_G=P\sum_{n=1}^{N}1/n$. The definition of Euler-Mascheroni constant is given as $\gamma= \sum_{n=1}^{N}1/n-\mathrm{ln}(N)$ with $N \rightarrow \infty$~\cite{MORTICI2010,lagarias2013}. Hence, $R_{N-1}(P)$ for genie aided method converges as \eqref{eq:SISOasym_gen}.
This completes the proof.
\end{proof}

By comparing \eqref{eq:SISOasym} and \eqref{eq:SISOasym_gen}, we can see that the cost of not being able to foresee the future is $\gamma$ for large $N$.

\subsection{With Channel Estimation Energy ($e_t \neq 0$)}
\label{sec:estimate}
When $e_t\neq 0$, the amount of energy used for channel estimation increases with the number of frames in which channel estimation is carried out. Hence, the available energy for energy transfer decreases with the number of frames. When the channel estimation energy, $e_t$, is significantly large, the available energy $P_{j,(available)}$ decreases rapidly with $j$. Then, at a certain frame $n$, the available energy  is not even adequate to support the channel estimation in next frame, i.e., $P_{n,(available)}<2e_t$. In such situations, all the available energy $P_{n,(available)}$ must be utilized for energy transfer at frame $n$ regardless of the number of frames $N$. This is similar to having an all-or-nothing threshold policy with zero threshold at frame $n$. By considering the previous frame ($n-1$), the available energy can be given as $P_{n-1,(available)} \geq P_{n,(available)} +e_t$. If the maximum harvestable energy in $n-1$ frame, $E_{n-1}=(P_{n-1,(available)}-e_t) |h_{n-1}|^2$, is higher than the expected harvested energy at the frame $n$, all the energy must be transmitted at frame $n-1$ in order to maximize the harvested energy. Thus, by conjecture, it is reasonable to  assume an all-or-nothing threshold policy, which is also motivated by Section~\ref{sec:Problem}. We further note that the thresholds when $e_t\neq 0$ case must be lower than the thresholds for corresponding frames when $e_t=0$ (Section~\ref{sec:Problem}).  Similar to \eqref{eq:EN(P)} by utilizing the recursive structure of the problem in \eqref{eq:OptE_sum_et}, we have
\begin{equation}\label{eq:frames_N}
\begin{aligned}
 \underset{{ \rho_{N-1} } }{\mathrm{max}} ~~~ & R_{N-1}(P)=\mathbb{E}_{h_{N-1}} \Big[{ \displaystyle   s_{N-1}(\rho_{N-1}-e_t)| h_{N-1} |^2} \\ 
  &{~~~~~+~ R_{N-2}(P-\rho_{N-1})  }\Big] \\
 \text{s.~t.} ~~~
 & e_t \leq \rho_{N-1} \leq P, ~R_{N-1}(P) \geq 0.
\end{aligned}
\end{equation}
Here, the constraint $R_{N-1}(P)\geq 0$ is imposed to handle case $P_{n,(available)}<2e_t$ which guarantees all available energy to transfer in frame $n$.

\subsubsection{Optimum Solution}
\label{sec:Problem2}
When the EH policy is based on an all-or-nothing threshold policy,  the optimum threshold policy and the optimum harvested energy is given in the  Theorem~\ref{t:th_et}.

\begin{theorem}\label{t:th_et}
With channel estimation energy ($e_t>0$), the optimum expected total 
harvested energy using $N$ frames over block fading Rayleigh channels is 
\begin{equation}\label{eq:Opt_et}
R_{N-1}(P)=s_{N-1}(P-e_t)\left[\gamma_{N-1}(P)+e^{-\gamma_{N-1}(P)}\right]
\end{equation}
where the optimum thresholds are 
\begin{equation}\label{eq:thres_et}
\gamma_j(P)=\left\{
\begin{array}{ll}
		\frac{P-2e_t}{P-e_t}\left[\gamma_{j-1}(P-e_t)+ e^{-\gamma_{j-1}(P-e_t)}\right], ~ P>2e_t\\
		0, ~~  \text{otherwise.}
	\end{array}
\right.
\end{equation}
\end{theorem}
 
\begin{proof}
We use method of induction to prove Theorem~\ref{t:th_et}. 

Base Case: For one frame ($N=1$) over Rayleigh fading, problem in   \eqref{eq:frames_N} can be given as
\begin{equation}\label{eq:OptE_sum_et2}
\begin{aligned}
R_{0}(P)=  \underset{~ \rho_{0}(x)}{\mathrm{max}}  &\int_o^\infty \left(\rho_{0}(x)-e_t\right)x e^{-x} dx   \\
 \text{s.~t.} ~~~
 & e_t \leq \rho_{0}(x) \leq P.
\end{aligned}
\end{equation}
It is obvious that the optimum transmit energy $\rho^*_0(x)$ is $P$, i.e., all the available energy is transmitted in the same frame regardless of the channel condition. Therefore, $R_0(P)=s_0(P-e_t)$ with threshold $\gamma_0(P)=0$.

Inductive hypothesis: Assume that the solution in \eqref{eq:Opt_et} is true for $N=n$. Then, we have
\begin{equation}\label{eq:Opt_etn}
R_{n-1}(P)=s_{n-1}(P-e_t)\left[\gamma_{n-1}(P)+e^{-\gamma_{n-1}(P)}\right]
\end{equation}
where the optimum thresholds are given by \eqref{eq:thres_et}.

Inductive Step: Consider the case $N=n+1$. The optimization problem in \eqref{eq:frames_N} can be written as
{\small\begin{equation}\label{eq:frames_Np1}
\begin{aligned}
 \underset{{ \rho_n } }{\mathrm{max}} ~~~ & R_n(P)=\mathbb{E}_{h_{n}} \Big[{ \displaystyle   s_n(\rho_n-e_t)| h_n |^2 + R_{n-1}(P-\rho_n)  }\Big] \\
 \qquad & \text{s.~t.} ~~~
  e_t \leq \rho_n \leq P, ~~ R_n(P) \geq 0.
\end{aligned}
\end{equation}}
By using the threshold policy assumption, we have two cases: i) $\rho_n=e_t$ for $| h_n |^2 < \gamma_n(P)$ which leads to $s_{n-1}=1$; and ii) $\rho_n=P$ for $| h_n |^2 > \gamma_n(P)$ which leads to $s_{n-1}=0$, where $\gamma_n(P)$ is the threshold at frame $n$. Thus, the objective function of \eqref{eq:frames_Np1} can be given as 
{\small\begin{equation}\label{eq:RN}
\begin{aligned}
& R_{n}(P)= s_n \bigg[(P-e_t)\Big(1+\gamma_n(P) \Big) e^{-\gamma_n(P)}~+ (P-2e_t)\\
&\Big(\gamma_{n-1}(P-e_t)+ e^{-\gamma_{n-1}(P-e_t)}\Big)
\Big(1-e^{-\gamma_n(P)}\Big)\bigg]
\end{aligned}
\end{equation}}

 Now we analyze the behavior of $R_{n}(P)$. The function in \eqref{eq:RN} has the form of $f:\mathbb{R}_+ \rightarrow \mathbb{R}_+$ with
$f(x)= a\left(1-e^{-x}\right)+b\left(1+x\right)e^{-x}$
where $a,b \in \mathbb{R}$. Then, we have
$\nabla f(x)= \left(a-bx\right)e^{-x}$.
As $e^{-x}>0$ for $x \in \mathbb{R}_+$, we have the following three cases. i) the function $f(x)$ is strictly increasing when $bx<a$ because $\nabla f(x)>0$; ii) the unique stationary point of $f(x)$ is obtained when $bx=a$; and iii) the function $f(x)$ is strictly decreasing when $bx>a$ because $\nabla f(x) <0$. This proves that the function $f(x)$ has a unique maximum. Furthermore, we note that strictly increasing/ decreasing functions are quasi-concave. Therefore, the function $f(x)$ is quasi-concave. This means $R_n(P)$ is quasi-concave. The Kuhn-Tucker-Lagrange (KTL) conditions~\cite{Arrow1961} for the problem in \eqref{eq:frames_Np1} can be given as 
\begin{subequations}\label{eq:KKT_R1}
\begin{alignat}{4}
 ( 1 + \lambda) \nabla_{\gamma_n} R_n(P) \leq 0   \label{eq:l1}\\
 \gamma_n ( 1 + \lambda) \nabla_{\gamma_n} R_n(P) =0 \label{eq:l2}\\
 \lambda \gamma_n = 0 \label{eq:l5}\\
 R_n(P) \geq 0, ~~\gamma_n \geq 0, ~~\lambda \geq 0 \label{eq:l6}
\end{alignat}
\end{subequations}
 By solving this system of equations for $\gamma_n(P)$ and $R_n(P)$, we have  \eqref{eq:Opt_et} and \eqref{eq:thres_et} for $N-1=n$. Based on the rule of induction, Theorem~\ref{t:th_et} is true for all $N > 0$.  This completes the proof.
\end{proof}
\section{Numerical and Simulation results}
\label{sec:MaxMinResults}

We use normalized energy unless otherwise specified, channels $h_j \sim \mathcal{CN}(0,1)$ and conversion efficiency $\eta =1$.  

\subsection{Without Channel Estimation Energy ($e_t=0$)}
\label{sec:MaxMinResultsnoet}

In this case, we sequentially compare $| h_j |^2$ with $\gamma_j$ in \eqref{eq:thresh_no_et}  and obtain optimum energy transfer $\rho^*_j$ in  \eqref{eq:cj}. Then, we calculate the total harvested energy in all $N$ frames. For the genie-aided energy transfer scheme, we choose the frame $j^*$ with maximum $| h_j |^2$ among the $N$ frames, and transmit all the energy available at that frame. For equal energy transfer, we transfer $P/N$ amount of energy in each frame. For the random energy transfer, we transfer all the energy available in a randomly selected frame $j \in \{0,N-1\}$. The average total harvested energy is calculated for $10^6$ channel realizations. 

Fig.~\ref{f:Energy_no_et} shows the variation of the expected harvested energy with the number of frames under four schemes.  The analytical results for opportunistic and genie-aided energy transfer schemes in \eqref{eq:Sol} and \eqref{eq:gen} match closely with the simulation results which verifies our analysis. The genie-aided scheme outperforms all other schemes, and optimum energy transfer scheme outperforms equal and random energy transfer schemes which have harvested energy as one. Compared to equal and random energy transfer schemes, the gain of expected harvested energy is approximately 4 times and 3.5 times for genie-aided energy transfer scheme and optimum EH scheme, respectively, at $N=30$.

Fig.~\ref{f:Converge_no_et} shows the variation of tern $R_{N-1}(P)-\ln (N)$ with the number of frames $N$. When $N$ increases, the simulated values approach to $\gamma$ and 0 for the genie-aided and the optimum energy transfer schemes, respectively. This verifies Theorem~\ref{t:Asym}. 

Fig.~\ref{f:Outage} shows the variation of the average outage probability with the transmit energy when the harvested energy threshold ($E_{Th}$) is 10dBm. We calculate the outage probability as $\mathbb{P}\displaystyle{\left(\tilde{E}_T < E_{Th} \right)}$. The Outage probabilities of genie-aided scheme improve as order $N$, i.e., order 1,2, and 3 for $N=1,2,$ and 3, respectively, because genie-aided scheme makes the decision based on $N$ independent statistics which helps for the diversity gain. Outage probabilities of optimum energy transfer scheme improve as order one irrespective of $N$, because its transmission is based on current frame only. However, it helps to provide an array gain, i.e., outage probability of approximately 0.1 is achieved with $P=20,~18,$ and 16dBm for N=1,2, and 3, respectively. Outage probabilities of all other cases improve also as order one, but without any array gain.
\begin{figure}
  \centering
  \includegraphics[width=1\figwidth]{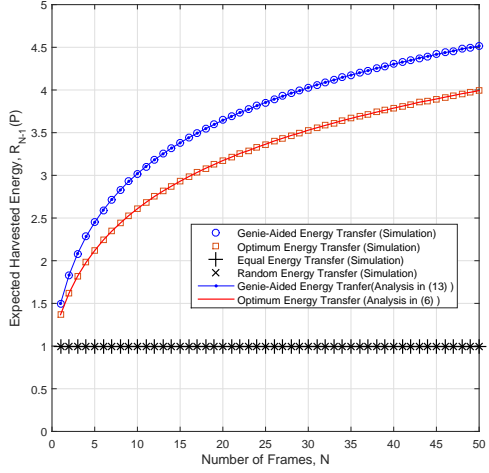}
\caption{Expected harvested energy with  number of frames for four energy transfer schemes when $e_t=0$.}
  \label{f:Energy_no_et}
\end{figure}
\begin{figure}
  \centering
  \includegraphics[width=1\figwidth]{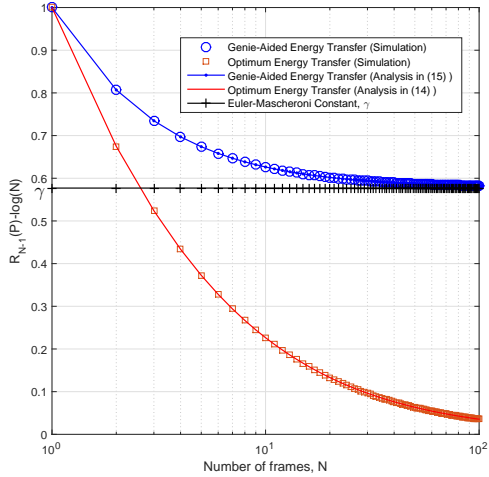} \\
\caption{ Asymptotic behavior of expected harvested energy with opportunistic and genie-aided schemes when $e_t=0$. }
  \label{f:Converge_no_et}
\end{figure}
\begin{figure}
  \centering
  \includegraphics[width=1\figwidth]{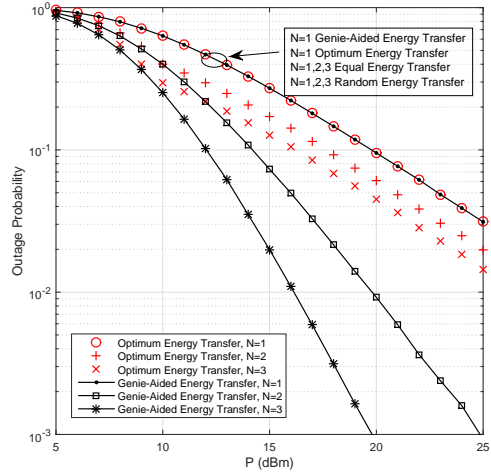} 
\caption{ Outage Probability of expected harvested energy  with transmit power $P$ for the energy threshold 10\,dBm and $e_t=0$.}
 \label{f:Outage}
\end{figure}

\subsection{With Channel Estimation Energy ($e_t > 0$)}
\label{sec:MaxMinResultset}

For $e_t >0$ case, in frame $N-1$, we compare $| h_{N-1}|^2$ with $\gamma_{N-1}(P)$ in \eqref{eq:thres_et} and obtain the optimum energy transfer $\rho^*_{N-1}$ in \eqref{eq:thres_et}. If  $\rho^*_{N-1}=e_t$, we update the available energy at next frame ($N-2$) as $P_{(N-2,available)}=1-e_t$. This procedure is repeated for all frames $N$.  For the genie aided energy transfer scheme, we calculate energy available for harvesting at frame $j$ as $P_{j,(available)}=1-(N-j-1)e_t$ and choose the frame $j^*$  with maximum $(P_{j,(available)} -e_t) |h_j |^2$ among the $N$ frames, and transmit available energy $P_{j^*,(available)}$ at this frame.  For the equal energy transfer, we transfer $P/N-e_t$ amount of energy in each frame.  For the random energy transfer, we transfer all the energy available, $P-e_t$ in a randomly selected frame $j \in \{0,N-1\}$. We calculate the average total harvested energy by using $10^5$ set of channel realizations.
\begin{figure}
  \centering
  \includegraphics[width=1\figwidth]{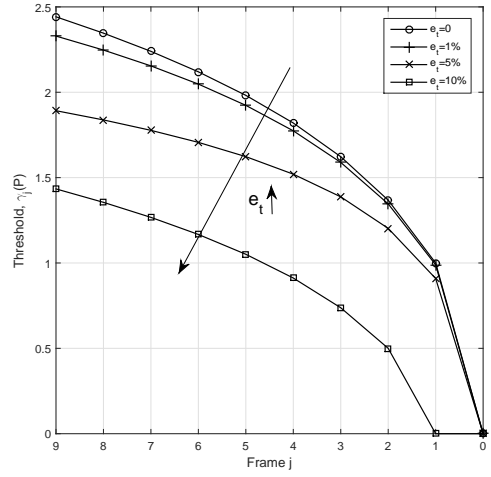}
\caption{ Threshold of frame $j$, $\gamma_j(P)$, with frame index $j$ for different $e_t$ values when $N=10$. }
  \label{f:Thresholds} 
\end{figure}
Fig.~\ref{f:Thresholds} shows the variation of the thresholds, $\gamma_j(P)$ for each frame with $N=10$ and $e_t=0,~0.01P,~0.05P$ and $0.10P$. For a given frame number, the threshold reduces when $e_t$ increases because the energy available for the future frame reduces significantly with large $e_t$. Therefore, harvesting at the current frame which may have a worse channel gain may result in large amount of harvested energy than harvesting in the future frame which may have a better channel gain but with low available energy. Furthermore, the threshold $\gamma_0(P)=0$ for all cases of $e_t$ because all available energy must be harvested at the last frame regardless of the channel. The threshold $\gamma_1(P)=0$ for frame 1 with $e_t=0.10P$ demonstrates the effect of $R_n(P)\geq 0$ constraint in (\ref{eq:frames_Np1}). As the available energy $P_{1,(available)} = 2e_t$, all available energy must be transmitted at frame 1 regardless of the channel quality. Hence the threshold $\gamma_1(P)$ is zero. 

Fig.~\ref{f:et_10_percent} shows the variation of expected harvested energy with the number of frames $N$ for four schemes when  $e_t=0.10P$. When $N=1$, the expected harvested energy from all four schemes is $0.9$. As there is only one frame, all the energy available, $1-e_t=0.9$, is transmitted in all four schemes. The harvested energy converges from 1 to 1.67 and from 1 to 1.50 for genie-aided and optimum energy transfer schemes, respectively, because no energy is available even for channel estimation when $N>10$ at $e_t=0.10P$. The harvested energy with the equal energy transfer scheme  drops to zero  because the energy utilized for energy transfer decreases as all the frames are used. When $N=10$, all the energy is utilized for channel estimation leaving no energy for energy transfer. Thus, the expected harvested energy for $N=10$, $E_9(P)$ is zero for equal energy transfer scheme. The expected harvested energy with random energy transfer scheme remains constant at $0.9$  because only one frame is utilized for energy transfer. Thus, the random energy transfer scheme outperforms the equal energy transfer scheme. 
\begin{figure}
  \centering
  \includegraphics[width=1\figwidth]{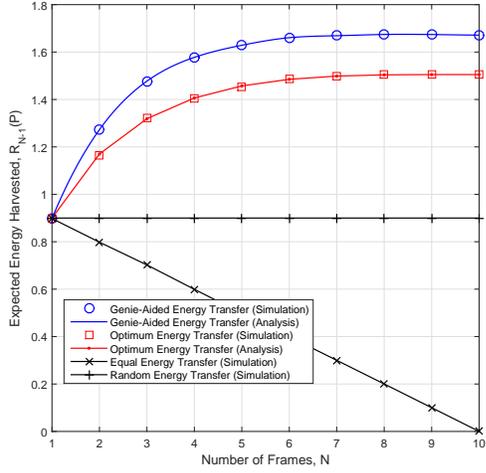} 
\caption{ Expected harvested energy with number of frames for four energy transfer schemes when $e_t=0.10P$ }
  \label{f:et_10_percent}
\end{figure}



\section{Conclusion}
\label{sec:Conclusion}

This paper considers a 
network with an energy constrained power transmitter and an EH user which has strict deadline on time. We obtain an opportunistic energy transfer scheme to maximize the expected total harvested energy without channel estimation energy. We show that the optimum policy is a threshold policy. We also derive the maximum expected total harvested energy using the genie-aided scheme for performance comparison. By analyzing the asymptotic behavior of two schemes, we show that the gain of non-causality can be given by the Eular-Mascheroni constant. Based on a threshold policy, we solve the problem with channel estimation energy.
We compare the performance using genie-aided energy transfer, equal energy transfer and random energy transfer schemes. We observe that the thresholds for each frame reduces with large channel estimation energy.



\newpage
\input{output-B.bbl}



\end{document}